\documentclass[11pt]{article}
\usepackage{fullpage}
\newtheorem{theorem}{Theorem}[section]
\newtheorem{corollary}[theorem]{Corollary}
\newtheorem{lemma}[theorem]{Lemma}

\newtheorem{definition}[theorem]{Definition}
\newtheorem{claim}[theorem]{Claim}
\newtheorem{remark}[theorem]{Remark}

\newenvironment{proof}{\noindent{\bf Proof : \ }}{\hfill$\Box$\par\medskip}

\usepackage{makeidx}
\usepackage[nointegrals]{wasysym}
\usepackage{amsmath}
\usepackage{graphicx}
\usepackage{fancybox}
\usepackage{url}
\usepackage{hyperref}
\usepackage{underscore}
\usepackage{tikz}
\usetikzlibrary{shadings}
\usetikzlibrary{decorations.pathreplacing}
\usepackage{setspace}
\usepackage{algorithm}
\usepackage[noend]{algpseudocode}
\usepackage{algorithmicx}
\usepackage[framemethod=tikz]{mdframed}
\usepackage{xspace}
\usepackage{pgfplots}
\usepackage{framed}
\pgfplotsset{compat=1.5}
\usepackage{standalone}

\usepackage[final]{changes} 
\newcommand{\stkout}[1]{\ifmmode\text{\sout{\ensuremath{#1}}}\else\sout{#1}\fi}
\setdeletedmarkup{\stkout{#1}}

\newtheorem{construction}[theorem]{Construction}

\newenvironment{proofof}[1]{\begin{trivlist} \item {\bf Proof
#1:~~}}
  {\qed\end{trivlist}}

\newcommand{\namedref}[2]{\hyperref[#2]{#1~\ref*{#2}}}
\newcommand{\thmlab}[1]{\label{thm:#1}}
\newcommand{\thmref}[1]{\namedref{Theorem}{thm:#1}}
\newcommand{\lemlab}[1]{\label{lem:#1}}
\newcommand{\lemref}[1]{\namedref{Lemma}{lem:#1}}
\newcommand{\claimlab}[1]{\label{claim:#1}}
\newcommand{\claimref}[1]{\namedref{Claim}{claim:#1}}
\newcommand{\corlab}[1]{\label{cor:#1}}
\newcommand{\corref}[1]{\namedref{Corollary}{cor:#1}}

\newcommand{\remlab}[1]{\label{rem:#1}}

\newcommand{\figlab}[1]{\label{fig:#1}}
\newcommand{\figref}[1]{\namedref{Figure}{fig:#1}}
\newcommand{\alglab}[1]{\label{alg:#1}}
\renewcommand{\algref}[1]{\namedref{Algorithm}{alg:#1}}

\newcommand{\deflab}[1]{\label{def:#1}}

\usetikzlibrary{decorations.pathreplacing}
\usetikzlibrary {positioning,chains,fit,shapes,calc,snakes}

\tikzset{
    position/.style args={#1:#2 from #3}{
        at=(#3.#1), anchor=#1+180, shift=(#1:#2)
    }
}

\definecolor{purduedigitalheadlinegold}{HTML}{98700D}
\definecolor{purduecampusgold}{HTML}{C28E0E}
\definecolor{purduecoalgray}{HTML}{4D4038}
\definecolor{purdueevertrueblue}{HTML}{5B6870}
\definecolor{purduemoondustgray}{HTML}{BAA892}
\definecolor{purdueslayterskyblue}{HTML}{6E99B4}
\definecolor{mahogany}{rgb}{0.75, 0.25, 0.0}
\definecolor{darkblue}{rgb}{0.0, 0.0, 0.55}
\definecolor{darkpastelgreen}{rgb}{0.01, 0.75, 0.24}
\definecolor{darkgreen}{rgb}{0.0, 0.2, 0.13}
\definecolor{darkgoldenrod}{rgb}{0.72, 0.53, 0.04}
\definecolor{forestgreen}{rgb}{0.13, 0.55, 0.13}
\definecolor{darkred}{rgb}{0.55, 0.0, 0.0}
\definecolor{tangocolordarkchameleon}{HTML}{4E9A06}

\newcommand{\COMMENTED}[1]{{}}

\newcommand{\nbr}[1]{\ensuremath{\mathsf{N}(#1)}}

\def \indeg    {\mdef{\textsf{indeg}}}

\def \path    {\mdef{\textsf{PATH}}}

\def \parents    {\mdef{\textsf{parents}}}

\renewcommand{\O}[1]{\ensuremath{\mathcal{O}\left(#1\right)}}

\newcommand{\eps}{\epsilon}

\def \GetParentsEGS    {\mdef{\mathsf{GetParentsEGS}}}

\renewcommand{\figurename}{Fig.}

\newcommand{\mdef}[1]{{\ensuremath{#1}}\xspace}  
\newcommand{\myfunc}[1]{\mdef{\mathsf{#1}}}      

\DeclareMathOperator*{\polylog}{polylog}

\newcommand{\superscript}[1]{\ensuremath{^{\mbox{\tiny{\textit{#1}}}}}\xspace}
\def \th {\superscript{th}}     

\def \polylog  {\mdef{\myfunc{polylog}\,}}             
\newcommand{\abs}[1]{\mdef{\left|#1\right|}}         
\newcommand{\flr}[1]{\mdef{\left\lfloor#1\right\rfloor}}              
\newcommand{\ceil}[1]{\mdef{\left\lceil#1\right\rceil}}               

\def \Trunc {\mdef{\textsc{Trunc}}}

\newcommand{\ignore}[1]{}

\newif\ifnotes\notestrue
\ifnotes
\newcommand{\seunghoon}[1]{\textcolor{purple}{{\bf (Seunghoon:} {#1}{\bf ) }} \marginpar{\tiny\bf
             \begin{minipage}[t]{0.5in}
               \raggedright S:
            \end{minipage}}}            
\newcommand{\jeremiah}[1]{\textcolor{red}{{\bf (Jeremiah:} {#1}{\bf ) }} \marginpar{\tiny\bf
             \begin{minipage}[t]{0.5in}
               \raggedright S:
            \end{minipage}}}
\newcommand{\mike}[1]{\textcolor{blue}{{\bf (Mike:} {#1}{\bf ) }} \marginpar{\tiny\bf
             \begin{minipage}[t]{0.5in}
               \raggedright S:
            \end{minipage}}}

\else
\newcommand{\samson}[1]{}
\newcommand{\jeremiah}[1]{}
\fi

\makeatletter
\renewcommand*{\@fnsymbol}[1]{\textcolor{mahogany}{\ensuremath{\ifcase#1\or *\or \dagger\or \ddagger\or
 \mathsection\or \triangledown\or \mathparagraph\or \|\or **\or \dagger\dagger
   \or \ddagger\ddagger \else\@ctrerr\fi}}}
\makeatother

\providecommand{\email}[1]{\href{mailto:#1}{\nolinkurl{#1}\xspace}}

\definecolor{mahogany}{rgb}{0.75, 0.25, 0.0}
\definecolor{darkblue}{rgb}{0.0, 0.0, 0.55}
\definecolor{darkpastelgreen}{rgb}{0.01, 0.75, 0.24}
\definecolor{darkgreen}{rgb}{0.0, 0.2, 0.13}
\definecolor{darkgoldenrod}{rgb}{0.72, 0.53, 0.04}
\definecolor{darkred}{rgb}{0.55, 0.0, 0.0}
\hypersetup{
     colorlinks   = true,
     citecolor    = mahogany,
		 linkcolor		= olive
}

\author{
Jeremiah Blocki\thanks{Department of Computer Science, Purdue University, West Lafayette, IN. 
Email: \email{jblocki@purdue.edu}}
\and
Mike Cinkoske\thanks{Department of Computer Science, University of Illinois at Urbana-Champaign, Urbana, IL.
Email: \email{mjc18@illinois.edu}}
\and
Seunghoon Lee\thanks{Department of Computer Science, Purdue University, West Lafayette, IN.
Email: \email{lee2856@purdue.edu}}
\and 
Jin Young Son \thanks{Department of Computer Science, Purdue University, West Lafayette, IN. 
Email: \email{son74@purdue.edu}}
}

\begin{document}
\title{On Explicit Constructions of Extremely Depth Robust Graphs}
\date{\today}

\maketitle
\begin{abstract}
A directed acyclic graph $G=(V,E)$ is said to be $(e,d)$-depth robust if for every subset $S \subseteq V$ of $|S| \leq e$ nodes the graph $G-S$ still contains a directed path of length $d$. If the graph is $(e,d)$-depth-robust for any $e,d$ such that $e+d \leq (1-\eps)|V|$ then the graph is said to be $\epsilon$-extreme depth-robust. In the field of cryptography, (extremely) depth-robust graphs with low indegree have found numerous applications including the design of side-channel resistant Memory-Hard Functions, Proofs of Space and Replication and in the design of Computationally Relaxed Locally Correctable Codes. In these applications, it is desirable to ensure the graphs are locally navigable, i.e., there is an efficient algorithm $\mathsf{GetParents}$ running in time $\polylog |V|$ which takes as input a node $v \in V$ and returns the set of $v$'s parents. We give the first explicit construction of locally navigable $\epsilon$-extreme depth-robust graphs with indegree $O(\log |V|)$. Previous constructions of $\epsilon$-extreme depth-robust graphs either had indegree $\tilde{\omega}(\log^2 |V|)$ or were not explicit.    
\end{abstract}
\section{Introduction}
A depth-robust graph $G=(V,E)$ is a directed acyclic graph (DAG) which has the property that for any subset $S \subseteq V$ of at most $e$ nodes the graph $G-S$ contains a directed path of length $d$, i.e., there is a directed path $P=v_0,\ldots, v_d$ such that $(v_i,v_{i+1}) \in E$ for each $i <d$ and $v_i \in V \setminus S$ for each $i \leq d$. As an example, the complete DAG $K_N=(V=[N], E=\{(i,j) : 1\leq i < j \leq n\}$ has the property that it is $(e,d)$-depth-robust for any integers $e,d$ such that $e+d \leq N$. Depth-robust graphs have found many applications in cryptography including the design of data-independent Memory-Hard Functions~(e.g.,~\cite{C:AlwBlo16,EC:AlwBloPie17}), Proofs of Space~\cite{C:DFKP15}, Proofs of Replication~\cite{ITCS:Pietrzak19a,EC:Fisch19} and Computationally Relaxed Locally Correctable Codes~\cite{ITIT:BGGZ21}. In many of these applications it is desirable to construct depth-robust graphs with low-indegree (e.g., $\indeg(G)=O(1)$ or $\indeg(G)=O(\log N)$) and we also require that the graphs are \emph{locally navigable}, i.e., given any node $v \in V=[N]$ there is an efficient algorithm $\mathsf{GetParents}(v)$ which returns the set $\{u : (u,v) \in E\}$ containing all of $v$'s parent nodes in time $O(\polylog N)$. It is also desirable that the graph is $(e,d)$-depth robust for $e,d$ as large as possible, e.g., the cumulative pebbling cost of a graph can be lower bounded by the product $ed$ and in the context of Memory-Hard Functions we would like to ensure that the cumulative pebbling cost is as large as possible~\cite{STOC:AlwSer15,EC:AlwBloPie17}. Some cryptographic constructions rely on an even stronger notion  called \emph{$\epsilon$-extreme depth-robust graphs} $G=(V,E)$ which have the property of being $(e,d)$-depth-robust for any integers $e,d$ such that $e+d \leq (1-\epsilon)N$, e.g., see \cite{ITCS:Pietrzak19a,ITCS:MahMorVad13}. 

Erd\"os, Graham, and Szemeredi \cite{EGS75} gave a randomized construction of $(e,d)$-depth-robust graphs with $e,d=\Omega(N)$ and maximum indegree $O(\log N)$. Alwen, Blocki, and Harsha~\cite{CCS:AlwBloHar17} modified this construction to obtain a locally navigable construction of $(e,d)$-depth-robust graphs with constant indegree $2$ for $e=\Omega(N/\log N)$ and $d=\Omega(N)$. For any constant $\epsilon>0$, Schnitger~\cite{FOCS:Schnitger83} constructed $(e=\Omega(N), d= \Omega(N^{1-\epsilon}))$-depth-robust graphs with constant indegree --- the indegree $\indeg(G)$ does increase as $\epsilon$ gets smaller. These results are essentially tight as {\em any} DAG $G$ \replaced{which}{with} is $\left( \frac{N \cdot i \cdot \indeg(G) }{\log N} , \frac{N}{2^i} \right)$-reducible\footnote{If a DAG $G$ is not $(e,d)$-depth-robust we say that it is $(e,d)$-reducible, i.e., there exists some set $S \subseteq V$ \added{of size $e$} such that $G-S$ contains no directed path of length $d$.} for any $i \geq 1$~\cite{C:AlwBlo16,Valiant77}. If $\indeg(G)=o(\log N)$ then the graph cannot be $(e,d)$-depth robust with $e,d = \Omega(N)$ and similarly if $\indeg(G)=\Theta(1)$ plugging in $i=O(\log \log N)$  demonstrates that $G$ cannot be $(e=\omega(N \log \log N/\log N), d= \omega(N))$-depth-robust. 

\paragraph*{Explicit Depth-Robust Graphs.} All of the above constructions are randomized and do not yield explicit constructions of depth-robust graphs. For example, the DRSample construction of \cite{CCS:AlwBloHar17} actually describes a randomized distribution over graphs and proves that a graph sampled from the distribution is $(e,d)$-depth-robust with high probability. Testing whether a graph is actually $(e,d)$-depth-robust is computationally intractable \cite{FC:BloZho18,ITCS:BloLeeZho20} so we cannot say that a particular sampled graph is depth-robust with $100\%$ certainty. In fact, it might be possible for a dishonest party to build a graph $G=(V,E)$ which looks like an honestly sampled depth-robust graph but actually contains a small (secret) depth-reducing set $S \subseteq V$, i.e., such that $G-S$ does not contain any long paths. Thus, in many cryptographic applications one must assume that the underlying depth-robust graphs \replaced{were}{was} generated honestly.

Li \cite{Li2019} recently gave an explicit construction of constant-indegree depth-robust graphs, i.e., for any $\epsilon >0$, Li constructs a family of graphs $\{G_{N,\epsilon}\}$ such that each $G_{N,\epsilon}$ has  $N$ nodes, constant indegree, and is $(\Omega(N^{1-\epsilon}),\Omega(N^{1-\epsilon}))$-depth-robust. The construction of Li \cite{Li2019} is also locally navigable, but the graphs are not as depth-robust as we would like. Mahmoody, Moran, and Vadhan~\cite{ITCS:MahMorVad13} gave an explicit construction of an $\epsilon$-extreme depth-robust graph for any constant $\epsilon>0$ using the Zig-Zag Graph Product constructions of \cite{FOCS:ReiVadWig00}. However, the maximum indegree is as large as $\indeg(G) \leq \log^3 N$. Alwen, Blocki, and Pietrzak~\cite{EC:AlwBloPie18} gave a tighter analysis of \cite{EGS75} \deleted{yields }showing that the randomized construction of \cite{EGS75} yields $\epsilon$-extreme depth-robust graphs with $\indeg(G)=O(\log N)$ although their randomized construction is not explicit nor was the graph shown to be locally navigable. 

\subsection{Our Contributions} We give explicit constructions of $\epsilon$-extreme depth-robust graphs with maximum indegree $O(\log N)$ for any constant $\epsilon >0$ and we also give explicit constructions of $( e= \Omega(N/\log N), \allowbreak d= \Omega(N))$-depth-robust graphs with maximum indegree $2$. Both constructions are explicit and locally navigable. In fact, our explicit constructions also satisfy a stronger property of being $\delta$-local expanders. A $\delta$-local expander is a directed acyclic graph $G$ which has the following property:  for any $r, v \geq 0$ and any subsets  $X \subseteq  A=[v,v+r-1]$ and $Y\subseteq B = [v+r,v+2r-1]$ of at least $|X|,|Y| \geq \delta r$ nodes the graph $G$ contains an edge $(x,y)$ with $x \in X$ and $y \in Y$.  We remark that the construction of Computationally Relaxed Locally Correctable Codes \cite{ITIT:BGGZ21} relies on a family of $\delta$-local expanders which is a strictly stronger property than depth-robustness --- \added{for any $\epsilon>0$, there exists a constant $\delta>0$ such that any} $\delta$-local expander\deleted{s are also } automatically \added{becomes} $\epsilon$-extreme depth-robust~\cite{EC:AlwBloPie18}.

\subsection{Our Techniques} We first provide explicit, locally navigable, constructions of $\delta$-bipartite expander graphs with constant indegree for any constant $\delta>0$. A bipartite graph $G=((A,B),E)$ with $|A|=|B|=N$ is a $\delta$-\replaced{bipartite}{local} expander if for {\em any} $X \subseteq A$ and $Y \subseteq B$ of size $|X|,|Y| \geq \delta N$ the bipartite graph $G$ contains at least one edge $(x,y) \in E$ with $x \in X$ and $y \in Y$. The notion of a $\delta$-bipartite expander is related to, but distinct from, classical notions of a graph expansion, e.g., we say that $G$ is a\added{n} $(N,k,d)$-expander if $\indeg(G)\leq k$ and for every subset $X \subseteq A$  (resp. $Y \subseteq B$) \deleted{of size $|A| \leq N/2$ }we have $|\replaced{\nbr{X}}{N(X)}| \geq (1+d-d|X|/N) |X|$ (resp. $|\replaced{\nbr{Y}}{N(Y)}|\geq (1+d-d|Y|/N) |Y|$)\added{, where $\nbr{X}$ is defined to be all of the neighbors of $X$, i.e., $\nbr{X} \doteq \{y \in B: \exists x \in X~\mbox{s.t.}~(x,y) \in E\}$. (Notation: We use $\nbr{X}$ (resp. $N$) to denote the neighbors of nodes in $X$ (resp. number of nodes in a graph/bipartition).)} Erd\"os, Graham, and Szemeredi \cite{EGS75} argued that a random degree $k_\delta$ bipartite graph will be a $\delta$-\replaced{bipartite}{local} expander with non-zero probability where the constant $k_{\delta}$ depends only on $\delta$. As a building block, we rely on an explicit, locally navigable, construction of $(n=m^2, k=5, d= (2-\sqrt{3})/4)$-expander graphs for any integer $m$ due to Gabber and Galil~\cite{GabGal81}. For any constant $\delta >0$ we show how any $(N,k,d)$-expander graph $G$ with $d < 0.5$ and $k=\Theta(1)$ can be converted into a $\delta$-bipartite expander graph $G'$ with $N$ nodes and maximum indegree $\indeg(G') =\Theta(1)$. Intuitively, the construction works by ``layering\replaced{''}{"} $\ell = \Theta(1)$ copies of the  $(N,k,d)$-expander graph\added{s} and then ``compressing\replaced{''}{"} the layers to obtain a bipartite graph $G'$ with maximum indegree $k' \leq k^\ell$ --- paths from the bottom layer to the top layer are compressed to individual edges. 

The depth-robust graph construction of Erd\"os et al.~\cite{EGS75} uses $\delta$-bipartite expanders as a building block. By swapping out the randomized (non-explicit) construction of $\delta$-bipartite expanders with our explicit and locally navigable construction, we obtain a family of explicit and locally navigable depth-robust graphs. Furthermore, for any $\epsilon >0$ we can apply \added{the} analysis of Alwen et al. \cite{EC:AlwBloPie18} to obtain explicit constructions of $\epsilon$-extreme depth-robust graphs by \replaced{selecting}{select} the constant $\delta>0$ accordingly. Finally, we can apply a standard indegree reduction gadget of Alwen et al. \cite{EC:AlwBloPie17} to obtain an $\left( e= N/\log N, d=\Omega(N)\right)$-depth-robust graph with indegree $2$.

\section{Preliminaries}
We use $[N]=\{1,\ldots, N\}$ to denote the set of all integers between $1$ and $N$ and we typically use $V=[N]$ to denote the set of nodes in our graph. It is often convenient to assume that $N=2^n$ is a power of $2$. Given a graph $G=(V=[N],E)$ and a subset $S\subseteq [N]$ we use $G-S$ to denote the graph obtained by deleting all nodes in $S$ and removing any incident edges. Fixing a directed graph $G=(V=[N],E)$ and a node $v \in V$, we use $\parents(v)=\{u~:~(u,v) \in E\}$ to denote the parents of node $v$ and we let $\indeg(G) = \max_{v \in [N]} \left| \parents(v) \right|$ denote the maximum indegree of any node in $G$. We say a DAG $G$ is $(e,d)$-reducible if there exists a subset $S\subseteq [N]$ of $|S| \leq e$ nodes such that $G-S$ contains no directed path of length $d$. If $G$ is not $(e,d)$-reducible we say that $G$ is $(e,d)$-depth-robust. 

We introduce the notion of a $\delta$-bipartite expander graph where the concept was first introduced by \cite{EGS75} and used as a building block to construct depth-robust graphs. Note \added{that }the specific name ``$\delta$-bipartite expander'' was not used in \cite{EGS75}. We follow the notation of~\cite{CCS:AlwBloHar17,EC:AlwBloPie18}. 

\begin{definition}\deflab{delta-bipartite}
A directed bipartite graph $G=((A,B),E)$ with $|A|=|B|=N$ is called a \emph{$\delta$-bipartite expander} if and only if for any subset $X\subseteq A,Y\subseteq B$ of size $|X|\geq \delta N$ and $|Y|\geq \delta N$ there exists an edge between $X$ and $Y$.
\end{definition}

\begin{remark}\remlab{delta-bip-property}
Observe that if $G=((A,B),E)$ is a $\delta$-bipartite expander then for any subset $X\subseteq A$ with $|X|\geq\delta N$ we must have $|\replaced{\nbr{X}}{N(X)}|\replaced{>}{\geq} (1-\delta)N$ where $\replaced{\nbr{X}}{N(X)} = \{y \in B: \exists x \in X~\mbox{s.t.}~(x,y) \in E\}$ denotes the neighbors of $X$. If this were not the case then we could take $Y= B \setminus \replaced{\nbr{X}}{N(X)}$ and we have $|Y| \geq \delta N$ and, by definition of $Y$, we have no edges between $X$ and $Y$ contradicting the assumption that $G$ is a $\delta$-bipartite expander. 
\end{remark}

\begin{definition}\deflab{nkd-expander}
A directed bipartite graph $G=((A,B),E)$ with $|A|=|B|=N$ is called an \emph{$(\replaced{N}{n},k,d)$-expander} if $|E|\leq k\replaced{N}{n}$ and for every subset $X\subseteq A$ (resp. $Y \subseteq B$) we have $|\replaced{\nbr{X}}{N(X)}|\geq \left[ 1+d\left(1-\frac{|X|}{\replaced{N}{n}}\right)\right]|X|$ (resp. $|\replaced{\nbr{Y}}{N(Y)}|\geq \left[ 1+d\left(1-\frac{|Y|}{\replaced{N}{n}}\right)\right]|Y|$) where $\replaced{\nbr{X}}{N(X)}= \{ y \in B ~:~\exists x \in X\mbox{~s.t.}~(x,y) \in E \}$ (resp. $\replaced{\nbr{Y}}{N(Y)}= \{ x \in A ~:~\exists y \in B \mbox{~s.t.}~(x,y) \in E \}$).
\end{definition}

Gabber and Galil~\cite{GabGal81} gave explicit constructions of $(N=m^2,k=5,d=(2-\sqrt{3})/5)$-expanders. \lemref{nkd-then-delta} highlights the relationship between $\delta$-bipartite expanders and the more classical notion of $(N,k,d)$-expanders. 

\begin{lemma}\lemlab{nkd-then-delta}
Let $0 < d < 1$ and let $\delta= \frac{(d+2)-\sqrt{d^2+4}}{2d}$.  If a directed bipartite graph $G=((A,B),E)$ with $|A|=|B|=\replaced{N}{n}$ is an $(N,k,d)$-expander for $d<1$ then $G$ is a $\delta$-bipartite expander.
\end{lemma}

\begin{proof}
Consider an arbitrary subset $X\subseteq A$ with $|X|\geq\delta N$ and let $Y = B \setminus \replaced{\nbr{X}}{N(X)}$. We want to argue that $|Y| < \delta N$ or equivalently $|\replaced{\nbr{X}}{N(X)}| > (1-\delta)N$. Without loss of generality, we may assume that $|X|<N$ (otherwise we have $\replaced{\nbr{X}}{N(X)} = B$ since $|\replaced{\nbr{X}}{N(X)}| \geq (1+d(1-|X|/N)) |X| = |X| = N$). Since $G$ is an $(\replaced{N}{n},k,d)$-expander, we have that $|\replaced{\nbr{X}}{N(X)}|\geq \left[ 1+d\left(1-\frac{|X|}{\replaced{N}{n}}\right)\right]|X| = -\frac{d}{\replaced{N}{n}}|X|^2 + (d+1)|X|$.  Hence, for $N > |X|\geq \delta N$, we have that
\begin{align*}
|\replaced{\nbr{X}}{N(X)}| &\geq -\frac{d}{N}|X|^2 + (d+1)|X| \\     &> -\frac{d}{N}(\delta N)^2 + (d+1)\delta N  \\      &\geq (1-\delta)N,
\end{align*}
where the middle inequality follows from the observation that when $d<1$, the function $f(x)=-\frac{d}{N}x^2+(d+1)x$ is an increasing function over the range $0\leq x\leq N$ and the last inequality follows from the choice of  $\delta= \frac{(d+2)-\sqrt{d^2+4}}{2d}$ since  $d\replaced{\geq}{\leq}\frac{1-2\delta}{\delta-\delta^2}$. Now fixing an arbitrary subset $Y \subseteq B$ with $|Y| \geq \delta N$ and setting $X = A \setminus \replaced{\nbr{Y}}{N(Y)}$, a symmetric argument shows that $|X| < \delta N$. Thus, $G$ is a $\delta$-bipartite expander. 
\end{proof}

\section{Explicit Constructions of \texorpdfstring{$\delta$}{Delta}-Bipartite Expanders}

In this section, we give an explicit (locally navigable) construction of a $\delta$-bipartite expander graph for any constant $\delta>0$. As a building block, we start with an explicit construction of $(N=m^2,k=5,d=(2-\sqrt{3})/4)$-expander due to Gabber and Galil \cite{GabGal81}. Applying \lemref{nkd-then-delta} above this gives us a $\delta$-bipartite expander with $\delta \approx 0.492$ whenever $N=m^2$. To construct depth-robust graphs we need to construct $\delta$-bipartite expanders for much smaller values of $\delta$ and for arbitrary values of $N$, i.e., not just when $N=m^2$ is a perfect square. We overcome the first challenge by layering the $(N=m^2,k,d)$-expanders of \cite{GabGal81} to obtain $\delta$-bipartite expanders for arbitrary constants $\delta>0$ --- the indegree increases as $\delta$ approaches $0$.  We overcome the second issues simply by truncating the graph, i.e., if $G$ is a $\delta/2$-bipartite expander with \replaced{$2N$}{$N$} nodes then we can discard up to $N/2$ sources and $N/2$ sinks and the remaining graph will still be a $\delta$-expander. 

\subsection{Truncation}
By layering the $(N,k,d)$-expanders of Gabber and Galil \cite{GabGal81} we are able to obtain a family $\{G_{m,\delta}\}_{m=1}^\infty$ of $\delta$-bipartite expanders for any constant $\delta>0$ such that $G_m$ has $N=m^2$ nodes on each side of the bipartition and constant indegree. However, our constructions of depth-robust graphs will require us to obtain a family $\{H_{N,\delta}\}_{N =1}^\infty$ of $\delta$-bipartite expanders such that $H_{N,\delta}$ has $N$ nodes on each side of the bipartition and constant indegree. In this section, we show how the family $\{H_{N,\delta}\}_{N=1}^\infty$ can be constructed by truncating graphs from the family $\{G_{m,\delta}\}_{m=1}^\infty$. Furthermore, if the construction of $G_{m,\delta}$ is explicit and locally navigable then so is $H_{N,\delta}$. 

For each $N$ we define $m(N):=\min_{m:m^2\geq N}$ to be the smallest positive integer $m$ such that $m^2 \geq N$. We first observe that for all integers $N \geq 1$ we have $m(N)^2 \geq N \geq m(N)^2/2$.

\begin{claim} \claimlab{claim:mn}
For all $N \geq 1$ we have $m(N)^2 \geq N \geq m(N)^2/2$.
\end{claim}
\begin{proof}
The fact that $m(N)^2 \geq N$ follows immediately from the definition of $m(N)$. For the second part it is equivalent to show that $m(N)^2/N \leq 2$ for all $N \geq 1$. The ratio $m(N)^2/N$ is maximized when $N=(m-1)^2+1$ for some $m\geq 1$. Thus, it suffices to show that $\frac{m^2}{(m-1)^2 + 1} \leq 2$ for all $m \geq 1$ or equivalently $1+\frac{2(m-1)}{(m-1)^2+1} \leq 2$. The function $f(m)=    \frac{2(m-1)}{(m-1)^2+1} $  is maximized at $m=2$ in which case $f(2)=1$.  For all $m \geq 2$ we have $1+\frac{2(m-1)}{(m-1)^2+1}\deleted{=2} \leq 2$ and when $m=1$ we have $ 1+\frac{2(m-1)}{(m-1)^2+1}=1 \leq 2$ so the claim follows. 
\end{proof}

Suppose that for any constant $\delta > 0$ we are given an explicit locally navigable family $\{G_{m,\delta}\}_{m=1}^\infty$ of $\delta$-bipartite expanders with $G_{m,\delta} = ((A_{m,\delta}=\{X_1,\ldots,X_{m^2}\},B_{m,\delta}=\{Y_1,\ldots, Y_{m^2}\}\added{)}, E_{m,\delta})$ with edge set $E_{m,\delta} = \{ \replaced{(X_i,Y_j)}{(a_i,b_j)} ~: ~ i \in \mathsf{GetParents}(m,\delta, j) \wedge j \leq m^2\}$ defined by an algorithm $\mathsf{GetParents}(m,\delta, j)$. We now define the algorithm $\mathsf{GetParentsTrunc}(N,\delta, j) = \mathsf{GetParents}(m(N),\delta/2, j) \cap \{1,\ldots ,N\}$ and we define  $H_{m,\delta} = ((A_{N,\delta}'=\{a_1,\ldots,a_{N}\},B_{N,\delta}'=\{b_1,\ldots, b_{N}\}\added{)}, E_{N,\delta}')$  with edge set $E_{N,\delta}' = \{ (a_i,b_j) ~: ~ i \in \mathsf{GetParentsTrunc}(N,\delta, j) \wedge j \leq N\}$. Intuitively, we start with a $\delta/2$-bipartite expander $G_{m,\delta/2}$ with $N'=m(N)^2$ nodes on each side of the partition and drop $N'-N \leq \replaced{N'/2}{N/2}$ nodes from each side of the bipartition to obtain $H_{m,\delta}$. Clearly, if $\mathsf{GetParents}$ can be evaluated in time $O(\polylog m)$ then $\mathsf{GetParentsTrunc}$ can be evaluated in time $O(\polylog N)$. Thus, the family  $\{H_{N,\delta}\}_{N =1}^\infty$ is explicit and locally navigable. Finally, we claim that  $H_{m,\delta}$ is a $\delta$-bipartite expander.

\begin{lemma}\lemlab{drop}
Assuming that $G_{m,\delta}$ is a $\delta$-bipartite expander for each $m\geq 1$ and $\delta>0$, the graph $H_{m,\delta}$ is a $\delta$-bipartite expander for each $m\geq 1$ and $\delta>0$. 
\end{lemma}

\begin{proof}
Consider two sets $X \subseteq \{1,\ldots , N\}$ and $Y \subseteq \{1,\ldots, N\}$ and set $m=m(N)$. If $|X| \geq \delta N$ and $|Y| \geq \delta N$ then by \claimref{claim:mn} we have $|X| \geq (\delta/2) m^2$ and $|Y|  \geq (\delta/2) m^2$. Thus, since $G_{m,\delta/2}$ is a $\delta/2$-bipartite expander and $X,Y \subseteq \{1,\ldots, m^2\}$ there must be some pair $(i,j) \in X \times Y$ with $i \in \mathsf{GetParents}(m,\delta/2, j)$. Since $i \leq N$ we also have  $i \in \mathsf{GetParentsTrunc}(N,\delta, j) = [N] \cap \mathsf{GetParents}(m,\delta/2, j)$. Thus, the edge $(a_i,b_j)$ still exists in the truncated graph $H_{m,\delta}$.  It follows that $H_{m,\delta}$ is a $\delta$-bipartite expander. 
\end{proof}

In the remainder of this section, we will focus on constructing $G_{m,\delta}$. In the next subsection, we first review the construction of $(N=m^2,k=5,d=(2-\sqrt{3})/4)$-expanders due to  Gabber and Galil \cite{GabGal81}.

\subsection{Explicit $(\replaced{N}{n},k,d)$-Expander Graphs}
 Let $P_m \doteq \{0,1,\ldots,m-1\}\times\{0,1,\ldots,m-1\}$ be the set of pairs of integers $(x,y)$ with $0 \leq x,y \leq m-1$. We can now define the family of bipartite graphs $G_{m}=((A_m,B_m),E_m)$ where $A_m=\{X_{i,j}=(i,j): (i,j) \in P_m \}$ and $B=\{Y_{i,j}=(i,j): (i,j) \in P_m\}$. The edge set $E_m$ is defined using the following $5$ permuatations  on $P_m$: 
\begin{align*}
\sigma_0(x,y) &= (x,y),\\
\sigma_1(x,y) &= (x,x+y),\\
\sigma_2(x,y) &= (x,x+y+1),\\
\sigma_3(x,y) &= (x+y,y),\\
\sigma_4(x,y) &= (x+y+1,y),
\end{align*}
where the operation $+$ is modulo $m$. Now we can define the edge set $E_m$ as 
\[ E_m=\{(X_{i',j'},Y_{i,j}):\exists~ 0\leq k\leq 4\text{ such that }\sigma_k(i',j')=(i,j)\}. \]

Gabber and Galil \cite{GabGal81} proved that the graph $G_m$  is a $(N,k,d)$-expander with $N=m^2$ nodes on each side of the biparition ($A_m$ / $B_m$), $k=5$, and $d=(2-\sqrt{3})/4$.

It will be convenient to encode nodes using integers between $1$ and $N=m^2$ instead of pairs in $P_m\deleted{\times P_m}$. define $\mathtt{PairToInt}_m(x,y)=xm+y + 1$, a bijective function mapping pairs $(x,y) \in  \{0,1,\ldots,m-1\}\times\{0,1,\ldots,m-1\}$ to integers $\{1,\ldots, m^2\}$ along with the inverse mapping $\mathtt{IntToPair}_m(z) = \left( \lfloor \frac{z-1}{m} \rfloor   , (z-1) \mod m \right)$. We can then redefine the permutations over the set $\{1,\ldots, m^2\}$ as follows $\sigma_j'(z) = \mathtt{PairToInt}_m\left( \sigma_j\left(\mathtt{IntToPair}\added{_m}(z) \right) \right)$ and we can (equivalently) redefine $G_{m}=((A_m,B_m),E_m)$ where $A_m=\{X_1,\ldots,X_{m^2}\}$, $B_m= \{Y_1,\ldots,Y_{m^2}\}$ and $E_m = \{ (X_{i},Y_{j}) ~: ~1 \leq j \leq m^2 \wedge i \in \mathsf{GetParentsGG}(m, \replaced{j}{y}) \}$. Here, $\mathsf{GetParentsGG}(m, \replaced{j}{y}) = \{\sigma_0'(j),\sigma_1'(j), \sigma_2'(j),\sigma_3'(j),\sigma_4'(j)\}$.

\subsection{Amplification via Layering}
\replaced{Given}{Now} that we have constructed explicit $\delta$-bipartite expanders with constant indegree for a fixed $\delta>0$, we will construct explicit $\delta$-bipartite expanders with constant indegree for any arbitrarily small $\delta>0$. The construction is recursive. As our base case we define $G_{m}^0 = G_m =((A_m,B_m),E_m)$ where $A_m=\{X_1,\ldots,X_{m^2}\}$, $B_m= \{Y_1,\ldots,Y_{m^2}\}$ and $E_m = \{ (X_{i},Y_{j}) ~: ~1 \leq j \leq m^2 \wedge i \in \mathsf{GetParentsGG}(m, \replaced{j}{y}) \}$ as the $(N=m^2,k=5,d=(2-\sqrt{3})/4)$-expander of Gabber and Galil \cite{GabGal81} and we define $\mathsf{GetParentsLayered}^1(m,\replaced{j}{y}) = \mathsf{GetParentsGG}(m,\replaced{j}{y})$. We can then define  $G_{m}^{i+1}=((A_m,B_m),E_m^{i+1})$ where $A_m=\{X_1,\ldots,X_{m^2}\},\allowbreak B_m= \{Y_1,\ldots,Y_{m^2}\}$ and $E_m^{i+1} = \{ (X_{i},Y_{j}) ~: ~1 \leq j \leq m^2 \wedge i \in \mathsf{GetParentsLayered}^{i+1}(m, \replaced{j}{y}) \}$ where $\mathsf{GetParents\added{Layered}}^{i+1}(m,\replaced{j}{y}) = \bigcup_{\replaced{j}{y}' \in \mathsf{GetParentsGG}(m,\replaced{j}{y})} \mathsf{GetParents\added{Layered}}^i(m,\replaced{j}{y}')$. Intuitively, we can form the graph $G_m^{i}$ by stacking $i$ copies of the graph $G_m$ and forming a new bipartite graph by collapsing all of the intermediate layers. See \figref{fig:layered} for an illustration.

\begin{figure}[ht]
\centering
\includegraphics[width=\linewidth]{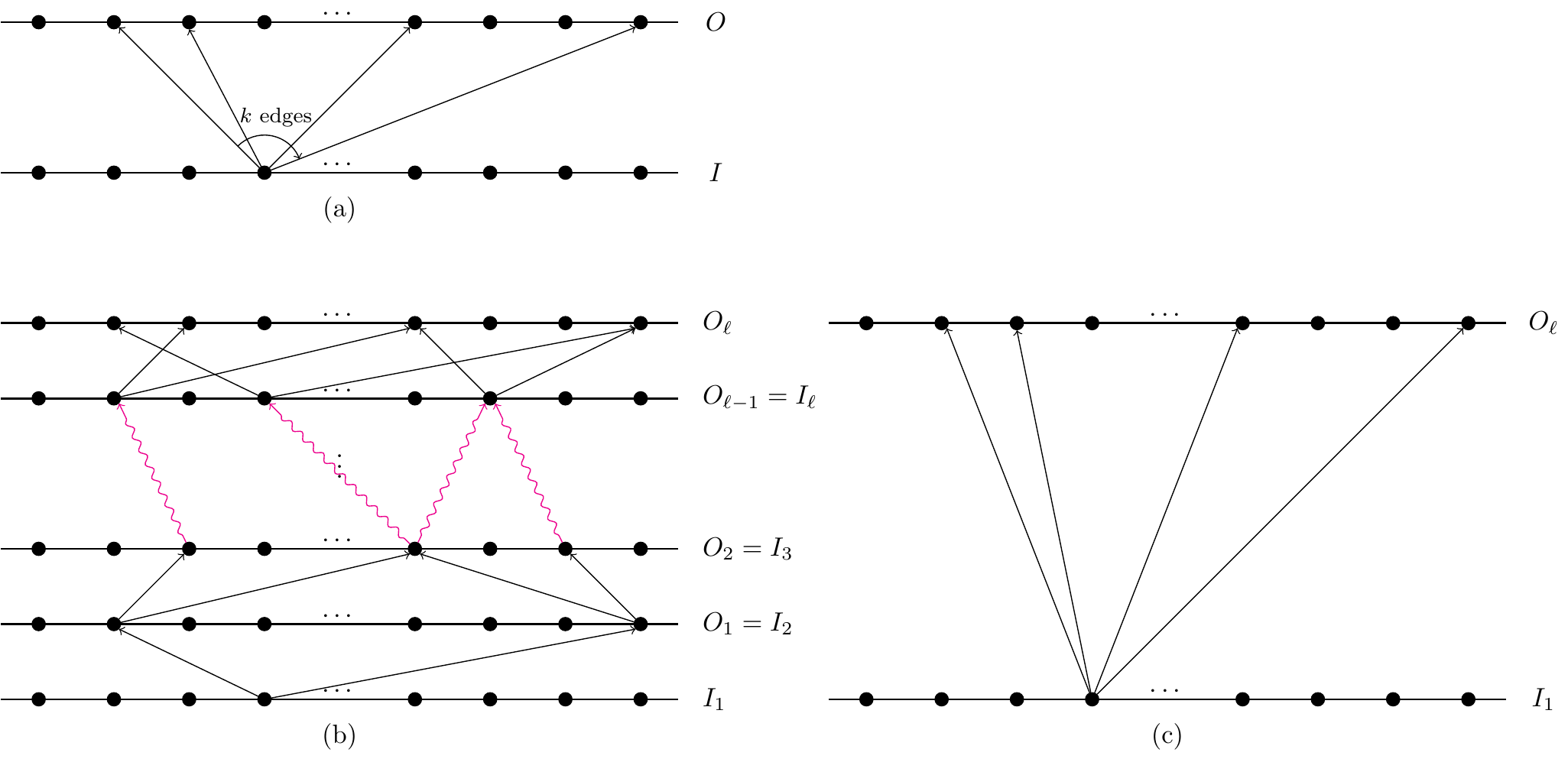}
\caption{(a) One copy of an $(\replaced{N}{n},k,d)$-expander. Here, we remark that each input node has exactly $k$ edges such that the total number of edges is \replaced{$kN$}{$kn$}. (b) Stack the graph $\ell$ times to get a graph with $(\ell+1)$ layers. The snaked edges from the third to $\ell$\th layer indicates that there are connected paths between the nodes. (c) Generate a new bipartite graph by collapsing all of the intermediate layers. A node $u$ on the bottom layer $I_1$ has an edge to a node $v$ on the top layer $O_{\ell}$ if and only if there is a path in the original graph.}
\figlab{fig:layered}
\end{figure}

We note that $\left| \mathsf{GetParents\added{Layered}}^{i+1}(m,\replaced{j}{y}) \right| \leq k \times  \left| \mathsf{GetParents\added{Layered}}^{i}(m,\replaced{j}{y}) \right| \leq k^{i+1}$. \deleted{Furthermore, $G_m^{i+1}$ is a $(N=m^2, k^{i+1}, d_{i+1})$-expander with $d_1=(2-\sqrt{3})/4)$ and $d_{i+1} $.} \added{\thmref{amplify} tells us that amplification by layering yields a $\delta$-bipartite expander. In particular, there is a constant $L_\delta$ such that $G_m^i$ is a $\delta$-bipartite expander whenever $i \geq L_{\delta}$. By our previous observation this graph has  indegree at most $k^{L_{\delta}}$ which is a constant since $k$ and $L_{\delta}$ are both constants.}

\begin{theorem} \thmlab{amplify}
\replaced{For any constant $\delta>0$, there exists a constant $L_\delta$ such that for any $i\geq L_\delta$ the graph $G_m^i$ is a $\delta$-bipartite expander with $N=m^2$ nodes on each side of the partition.}
{For each $i, m \geq 1$ the graph $G_m^{i}$ is a $\delta_i$-bipartite expander for some sequence $\{\delta_i\}_{i=1}^\infty$ such that  $\delta_{i+1} \leq \delta_{i}(1 - d) + \delta_{i}^2 d$ where  $d = \frac{2 - \sqrt{3}}{4}$ and  $\delta_1 = \frac{10 - \sqrt{3} - \sqrt{71-4\sqrt{3}}}{2(2-\sqrt{3})} \approx .4916$.} 
\end{theorem}

\begin{proof}
Fix any subset $Y^0\subseteq[N]$ of size $|Y^0|\geq\delta N$. Let $Y^1\doteq\bigcup_{j\in Y^0}\mathsf{GetParentsGG}(m,j)$, and recursively define $Y^{i+1}\doteq\bigcup_{j\in Y^i}\mathsf{GetParentsGG}(m,j)$. Since $Y^i=\bigcup_{j\in Y^0}\mathsf{GetParentsLayered}^i(m,j)$, it suffices to argue that $|Y^i|>(1-\delta)N$ whenever $i\geq L_\delta\doteq\left\lceil\frac{\log((1-\delta)/\delta)}{\log(1+d\delta)}\right\rceil + 1$. To see this, we note that for each $i\geq 0$, either 
\begin{enumerate}
\item $|Y^i|$ has already reached the target size $(1-\delta)N$, or
\item $|Y^{i+1}|\geq\left[1+d\left(1-\frac{|Y^i|}{N}\right)\right]|Y^i|\geq(1+d\delta)|Y^i|$ since $\mathsf{GetParentsGG}$ defines an $(N,k,d)$-expander.
\end{enumerate}
It follows that $|Y^{i+1}|\geq\min\{(1-\delta)N,(1+d\delta)^i\delta N\}$. Now we want to find $i$ such that $(1+d\delta)^i\delta N = (1-\delta)N$; solving the equation we have $i=\frac{\log((1-\delta)/\delta)}{\log(1+d\delta)}$. Thus, for $i = L_{\delta}-1$ we have $|Y^{i}| \geq (1-\delta)N$ and for $i \geq L_{\delta}$ we have $|Y^{i}| > (1-\delta)N$. Thus, for $i\geq L_\delta$ the graph $G_m^i$ is a $\delta$-bipartite expander, i.e., for any subsets $X,Y \subseteq [N]$ of size $|X| \geq \delta N = \delta m^2$ we must have $\left| X \cap \bigcup _{j \in Y}\mathsf{GetParentsLayered}^i(m,j)  \right| > 0$ as long as $i \geq L_\delta$.
\end{proof}

\ignore{
\begin{algorithm}
  \caption{$\mathsf{GetParents}(n, v)$ Finds the parents of node $v$ in $G_n$.}
  \label{getparents}
  \begin{algorithmic}[1]
    \Procedure{\textsf{GetParents}}{$n, v$} 
    \State $i \gets \lfloor v/m(n)\rfloor$
    \State $j \gets n - m(n)i$ 
    \State \textbf{return} $\{v, [m(n)i + (j - i)], [m(n)i + (j - i - 1)], [m(n)(i - j) + j], [m(n)(i - j - 1) + j] \}$ 
    \EndProcedure
  \end{algorithmic}
\end{algorithm}      

\begin{claim}
Let $\{\delta_i\}_{i=1}^\infty$ be any seqeuence with $\delta_1 < 0.5$ and $ \delta_{i+1} \leq \delta_{i}(1 - d) + \delta_{i}^2 d$ with $d = \frac{2 - \sqrt{3}}{4}$. Then for all $i \geq 1$ we have $0 \leq \delta_{i+1} \leq \delta_1 \left(1- \frac{d}{2} \right)^{i}$ and in particular $\lim_{i \rightarrow \infty} = 0$.  \claimlab{well-defined} 
\end{claim}
\begin{proof} 
Observe that if $\delta_i< 0.5 $ then $\frac{\delta_{i+1}}{\delta_{i}} <(1-d) + \delta_id < 1-d/2 < 1$. Thus, $\delta_{i+1} \leq \delta_1 \left(1- \frac{d}{2} \right)^{i}$ and it follows that   $0 \leq \lim_{i\to\infty} \delta_i  \leq \lim_{i\to\infty} \delta_1 \left(1- \frac{d}{2} \right)^{i-1} = 0$.  
\end{proof}
}

\subsection{Final Construction of $\delta$-Bipartite Expanders}

Based on the proof of \thmref{amplify}, we can define \replaced{$L_\delta\doteq \left\lceil\frac{\log((1-\delta)/\delta)}{\log(1+d\delta)}\right\rceil+1$}{$L_\delta = \lceil \frac{\ln (\delta/\delta_1)}{\ln (1-d/2)} \rceil$},  $G_{m, \delta} \doteq G_m^{L_{\delta}}$, and obtain $H_{N,\delta}$ by truncating the graph $G_{m(N), \delta/2}$. The edges are defined by the \replaced{procedure}{prodecure} $\mathsf{GetParentsBE}(N, \delta, j) \doteq [N] \cap \mathsf{GetParents\added{Layered}}^{L_{\delta/2}}(m(N), j)$ --- the procedure  $\mathsf{GetParentsBE}$ is short for ``Get Parents Bipartite Expander''. Formally, we have $H_{N,\delta} = ((A_N=\{a_1,\ldots, a_N\},\allowbreak B_N=\{b_1,\ldots, b_N\}), E_{N,\delta})$ where  $ E_{N,\delta}= \{ (a_i,b_j)~: ~i \in \mathsf{GetParentsBE}(N,\delta, j)\}$.

\begin{corollary}
Fix any constant $\delta > 0$ and define \replaced{$L_\delta = \left\lceil\frac{\log((1-\delta)/\delta)}{\log(1+d\delta)}\right\rceil+1$}{$L_\delta = \lceil \frac{\ln (\delta/\delta_1)}{\ln (1-d/2)} \rceil +1$}. The graph  $G_m^{L_\delta }$ is a $\delta$-bipartite expander and the graph $H_{N,\delta}$ is a $\delta$-bipartite expander for any integers $m,N \geq 1$. 
\end{corollary}
\begin{proof}
By \thmref{amplify} $G_m^{L_{\delta}}$ is a \replaced{$\delta$}{$\delta_{i}$}-bipartite expander\deleted{ with $i = L_\delta$}. \deleted{By (claim) we have \[ \delta_i \leq \delta_1   \left(1- \frac{d}{2} \right)^{ \lceil \frac{\ln (\delta/\delta_1)}{\ln (1-d/2)} \rceil} \leq \delta_1 \left( \frac{\delta}{\delta_1}\right) = \delta \ . \]}
To see that $H_{N,\delta}$ is a $\delta$-bipartite expander we simply note that $G_{m(N), \delta/2}$ is a $\delta/2$-bipartite expander and apply \lemref{drop}.
\end{proof}

\section{Explicit Constructions of Depth Robust Graphs}
We are now ready to present our explicit construction of a depth-robust graph. For any $N=2^n$ we define the graph $G(\delta, N) = ([N], E(\delta, N))$ with edge set $E(\delta, N) = \{ (u,v)~: ~ v \in [N] \wedge u \in \GetParentsEGS(\delta, v, N)\}$. The procedure $\GetParentsEGS(\delta, v, N)$ to compute the edges of $G(\delta,N)$ relies on the procedure $\mathsf{GetParentsBE}$ which computes the edges of our underlying bipartite expander graphs. We remark that our construction is virtually identical to the construction of \cite{EGS75} except that the underlying bipartite expanders are replaced with our explicit constructions from the last section.

\begin{algorithm}
    \caption{$\GetParentsEGS(\delta, v, N)$}
    \alglab{getparentsegs}
    \begin{algorithmic}[1]
    \Procedure{\textsf{GetParentsEGS}}{$\delta, v,N$}
	\State $P = \{v-4n, ... , v-1 \}$
    \For {$t=1$ to $\ceil{\log_2{v}}$}
    \State {$m = \flr{v/2^t}$}
    \State {$x = v \mod 2^t$}
    \State {$B = \mathsf{GetParentsBE}(2^t, \replaced{L_{\delta/5}}{l_{\delta/5}},x+1)$}
        \For {$y \in B$}
        \State {$P = P \cup \{(m-i)2^t + y\added{\,:\,1\leq i\leq \min\{m,10\}}\}$}
        \EndFor
    \EndFor
    \State \textbf{return} $P \cap \{1, ... , N\}$
    \EndProcedure
    \end{algorithmic}
\end{algorithm}

Note that for any constant $\delta > 0$ and any integer $n \geq 1$, the graph $G(\delta, N)$ defined by $\GetParentsEGS(\delta, \cdot, N )$ has $N=2^n$ nodes and maximum indeg $\indeg(G(\delta, N))= O(n) = O(\log N)$.

Erd\"os, Graham, and Szemeredi \cite{EGS75} showed that the graph $G(\delta , N)$ is a $\delta$-local expander as long as the underlying bipartite graphs are $\delta/5$-bipartite expanders. 

\begin{theorem}[\cite{EGS75}] \thmlab{EGSIsLocalExpander} For any $\delta >0 $ the graph  $G(\delta, N)$ is a $\delta$-local expander.
\end{theorem}

 \thmref{LocalExpandersAreDepthRobust} says that any $\delta$-local expander is also \deleted{depth-robust }$(e, d=N-e\frac{1+\gamma}{1-\gamma})$-depth-robust for any constant $\gamma > 2 \delta$. The statement of \thmref{LocalExpandersAreDepthRobust} is implicit in the analysis of Alwen et al. \cite{EC:AlwBloPie18}. We include the proof for completeness.

\begin{theorem} \thmlab{LocalExpandersAreDepthRobust}
    Let $0 < \delta < 1/4$ be a constant and let $\gamma > 2 \delta$. Any $\delta$-local expander on $N$ nodes is $(e, d=N-e\frac{1+\gamma}{1-\gamma})$-depth-robust for any $e \leq N$.
\end{theorem}
\begin{proof}
Let $G$ be a $\delta$-local expander with $\delta < 1/4$ and $\gamma> 2\delta$ and \deleted{consider and }let $S\subseteq [N]$ denote an arbitrary subset of size $|S|=e$. To show that $G-S$ has a path of length $d=N-e\frac{1+\gamma}{1-\gamma}$ we rely on two lemmas  (\lemref{gamma-good}, \lemref{depth-robust}) due to Alwen et al. \cite{EC:AlwBloPie18}. We first introduce the notion of a $\gamma$-good node. A node $x \in [N]$ is $\gamma$-good under a subset $S \subseteq [N]$ if for all $r > 0$ we have $|I_r(x)\backslash S| \geq \gamma |I_r(x)|$ and $|I_r^*(x)\backslash S| \geq \gamma |I_r^*(x)|$, where $I_r(x)=\{x-r-1, ..., x \}$ and $I_r^*(x)=\{x+1, ..., x+r \}$.

\begin{lemma}[\cite{EC:AlwBloPie18,EGS75}]\lemlab{gamma-good}
    Let $G=(V = [N], E)$ be a $\delta$-local expander and let $x<y\in [N]$ both be $\gamma$-good under $S\subseteq [N]$ then if $\delta < \min(\gamma/2, 1/4)$ then there is a directed path from node $x$ to node $y$ in $G-S$.
\end{lemma}
\begin{lemma}[\cite{EC:AlwBloPie18}]\lemlab{depth-robust}
    For any DAG $G=([N], E)$ and any subset $S\subseteq [N]$ of nodes at least $N-\abs{S}\frac{1+\gamma}{1-\gamma}$ of the remaining nodes in G are $\gamma$-good with respect to $S$.
\end{lemma}
Applying \lemref{depth-robust} at least  $d=N-e\frac{1+\gamma}{1-\gamma}$ nodes $v_1,\ldots, v_d$ \replaced{are}{which are all} $\gamma$-good with respect to $S$. Without loss of generality, we can assume that $v_1 < v_2 < \ldots < v_d$. Applying \lemref{gamma-good} for each $i \leq d$, there is a directed path from $v_i$ to $v_{i+1}$ in $G-S$. Concatenating all of these paths we obtain one long directed path containing all of the nodes $v_1,\ldots, v_d$. Thus, $G-S$ contains a directed path of length $d=N-e\frac{1+\gamma}{1-\gamma}$.
\end{proof}

As an immediate corollary of \thmref{EGSIsLocalExpander}  and  \thmref{LocalExpandersAreDepthRobust} we have

\begin{corollary} \corlab{DR}
    Let $0 < \delta < 1/4$ be a constant and let $\gamma > 2 \delta$ then the graph $G(\delta,N)$ is $(e, d=N-e\frac{1+\gamma}{1-\gamma})$-depth-robust for any $e \leq N$.
\end{corollary}

\subsection{Explicit Extreme Depth-Robust Graphs}


We also obtain explicit constructions of $\epsilon$-extreme depth-robust graphs which have found applications in constructing Proofs of Space and Replication \cite{ITCS:Pietrzak19a}, Proofs of Sequential Work~\cite{ITCS:MahMorVad13}, and in constructions of Memory-Hard Functions~\cite{EC:AlwBloPie18}.

\begin{definition} [\cite{EC:AlwBloPie18}] \deflab{extreme-dr}
    For any constant $\epsilon > 0$, a DAG $G$ with $N$ nodes is \emph{$\epsilon$-extreme depth-robust} if and only if $G$ is $(e,d)$-depth-robust for any $e+d \leq (1-\epsilon)N$.
\end{definition}

When we set $\delta_\epsilon$ appropriately the graph $G(\delta_\epsilon, N=2^n)$ is $\epsilon$-extremely depth robust. 

\begin{corollary}\corlab{extreme-dr}
Given any constant $\epsilon > 0$ we define $\delta_\epsilon$ to be the unique value such that $1+\epsilon = \frac{1+2.1\delta_\epsilon}{1-2.1\delta_\epsilon} $ if $\epsilon \leq 1/3$ and $\delta_{\epsilon} = \delta_{1/3}$ for $\epsilon > 1/3$. For any integer $n \geq 1$ the graph $G(\delta_\epsilon, N=2^n)$ is  $\epsilon$-extreme depth robust.
\end{corollary}
\begin{proof}
Set $\gamma=2.1  \delta_\epsilon$ and observe that $\delta_{1/3} \leq 0.07 \leq 1/4$ and for $\epsilon < 1/3$ we have $\delta_{\epsilon} \leq \delta_{1/3} \leq 1/4$ so we can apply \corref{DR} to see that $G(\delta_\epsilon, N=2^n)$  is $(e,d=N-e\frac{1+2.1 \delta_\epsilon}{1 - 2.1 \delta_\epsilon})$-depth robust for any $e \leq N$. Since  $\frac{1+2.1 \delta_\epsilon}{1 - 2.1 \delta_\epsilon} = (1+\epsilon)$ it follows that the graph is $\epsilon$-extreme depth robust.
\end{proof}

\subsection{Depth-Robust Graphs with Constant Indegree}
In some applications it is desirable to ensure that our depth-robust graphs have constant indegree. We observe that we can apply a result of Alwen et al.~\cite{EC:AlwBloPie17} to transform the DAG $G(\delta, N) = (V=[N], E(\delta,N))$ with maximum indegree $\beta=\beta_{\delta,N}$ into a new DAG $H_{\delta, N} = ([N] \times [\beta], E'(\delta, N))$ with $N'= 2N\beta$ nodes and maximum indegree $2$. Intuitively, the transformation reduces the indegree by replacing every node $v \in [N]$ from $G(\delta, N)$ with a path of $2 \beta$ nodes $(v,1),\ldots, (v,2\beta)$ and distributing the incoming edges accross this path. In particular, if $v$ has incoming edges from nodes $v_1,\ldots, v_{\beta}$ in $G(\delta, N)$ then for each $i \leq \beta$ we will add an edge from the node $(v_i,2\beta)$ to the node $(v, i)$. This ensures that each node $(v,i)$ has at most two incoming edges. Formally, the algorithm $\mathsf{GetParentsLowIndeg}(\delta, \replaced{v', N}{N, \cdot})$ takes as input a node $v'=(v,i)$ and (1) initializes $P'=\{(v,i-1)\added{\}}$ if $i > 1$,  $P'=\{(v-1,2\beta)\}$ if $i=1$ and $v>1$ and $P'=\{\}$ otherwise, (2) computes $P = \mathsf{GetParentsEGS}(\delta, \replaced{v,N}{N, v})$, (3) sets $u = P[i]$ to be the $i$th node in the set $P$, and (4) returns $P' \cup \{(u,2\beta)\}$. It is easy to verify that the algorithm $\mathsf{GetParentsLowIndeg}$ runs in time $\polylog N$.  

\begin{corollary} 
    Let $0 < \delta < 1/4$ be a constant and let $\gamma > 2 \delta$ then the graph $H_{\delta, N}$ is $(e, d=N\beta -e\beta \frac{1+\gamma}{1-\gamma})$ depth-robust for any $e \leq N$. 
\end{corollary}
\begin{proof}(Sketch) Alwen et al.~\cite{EC:AlwBloPie17} showed that applying the indegree reduction procedure above to any $(e,d)$-depth-robust graph with maximum indegree $\beta$ yields a $(e,d\beta)$-depth-robust graph. The claim now follows directly from \thmref{EGSIsLocalExpander}  and  \thmref{LocalExpandersAreDepthRobust}. \end{proof}

\section{Conclusion}

We give the first explicit construction of $\epsilon$-extreme depth-robust graphs $G=(V=[N],E)$ with indegree $O(\log N)$ which are locally navigable. Applying an indegree reduction gadget of Alwen et al.~\cite{EC:AlwBloPie17} we also obtain the first explicit and locally navigable construction of $\left( \Omega(N/\log N), \Omega(N)\right)$-depth-robust graphs with constant indegree. Our current constructions are primarily of theoretical interest and we stress that we make no claims about the practicality of the constructions as the constants hidden by the asymptotic notation are large. Finding explicit and locally navigable constructions of $(c_1 N/\log N, c_2 N)$-depth-robust graphs with small indegree for reasonably large constants $c_1, c_2 > 0$ is an interesting and open research challenge. Similarly, finding explicit and locally navigable constructions of  $\epsilon$-extreme depth-robust graphs $G=(V=[N],E)$ with indegree $c_\epsilon \log N$ for smaller constants $c_{\epsilon}$ remains an important open challenge.

\def\shortbib{0}
\bibliographystyle{alpha}
\bibliography{references,abbrev3,crypto}

\end{document}